\documentclass[5p, times]{elsarticle}

\usepackage{ecrc}

\volume{00}

\firstpage{1}

\journalname{Systems \& Control Letters}

\runauth{}
\jid{procs}
\jnltitlelogo{Systems \& Control Letters}
\CopyrightLine{2024}{Published by Elsevier Ltd.}

\usepackage{amssymb}
\usepackage{amsmath}
\usepackage{amsfonts}
\usepackage{xcolor}
\usepackage{amsthm}
\usepackage{mathtools}
\usepackage{appendix}
\usepackage{bm}
\usepackage{float}
\usepackage{hyperref}

\newtheorem{definition}{Definition}
\newtheorem{theorem}{Theorem}

\newtheorem{corollary}{Corollary}
\newtheorem{lemma}{Lemma}

\newtheorem{assumption}{Assumption}

\newcommand{\rplus}[0
]{\mathbb{R}_+}

\usepackage[figuresright]{rotating}

\begin{document}

\begin{frontmatter}


\dochead{}

\title{Adaptive control of reaction-diffusion PDEs\\ via neural operator-approximated gain kernels \tnoteref{funding}}
\tnotetext[funding]{The first author is supported by the U.S. Department of Energy (DOE) grant DE-SC0024386. The work of M. Krstic was funded by AFOSR grant FA9550-23-1-0535  and NSF grant ECCS-2151525.}


\author[eceDept]{Luke Bhan\corref{corresponding}}
\cortext[corresponding]{Corresponding Author}
\author[eceDept]{Yuanyuan Shi}
\author[maeDept]{Miroslav Krstic}

\address[eceDept]{Department of Electrical and Computer Engineering, University of California San Diego, La Jolla, CA 92093-0411, USA}
\address[maeDept]{Department of Mechanical and Aerospace Engineering, University of California San Diego, La Jolla, CA 92093-0411, USA}

\begin{abstract}
Neural operator approximations of the gain kernels in PDE backstepping has emerged as a viable method for implementing controllers in real time. With such an approach, one approximates the gain kernel, which maps the plant coefficient into the solution of a PDE, with a neural operator. It is in adaptive control that the benefit of the neural operator is realized, as the kernel PDE solution needs to be computed online, for every updated estimate of the plant coefficient. We extend the neural operator methodology from adaptive control of a hyperbolic PDE to adaptive control of a benchmark parabolic PDE (a reaction-diffusion equation with a spatially-varying and unknown reaction coefficient). We prove global  stability and asymptotic regulation of the plant state for a Lyapunov design of parameter adaptation. The key technical challenge of the result is handling the $2D$ nature of the gain kernels and proving that the target system with two distinct sources of perturbation terms, due to the parameter estimation error and due to the neural approximation error, is Lyapunov stable.  To verify our theoretical result, we present simulations achieving calculation speedups up to $45\times$ relative to the traditional finite difference solvers for every timestep in the simulation trajectory. 
\end{abstract}

\end{frontmatter}


\section{Introduction}
\label{sec:introduction}
First introduced in \cite{bhan_neural_2023}, a new methodology has emerged for employing \emph{neural operator}(NO) approximations of the gain kernels in PDE backstepping. The key advantage of this approach compared to traditional implementations of the kernel, as well as other approximation approaches such as \cite{WOITTENNEK20176786}, 
is the ability to produce the \textbf{entire} kernel in mere milliseconds for the online control law while the training process is decoupled to be precomputed offline. Therefore, perhaps the most valuable application of the neural operator approximated gain kernels is in adaptive control, where the kernel needs to be recomputed, online, for every new estimate of the plant parameter. This was first explored for hyperbolic PDEs in \cite{lamarque2024adaptive}. In this work, we extend the results of \cite{lamarque2024adaptive} to parabolic PDEs where the technical challenge arises both in the Lyapunov analysis of a more complex perturbed target system as well as in the computational implementation, where the neural operator must map functions on an $1D$ domain into a $2D$ triangle. Furthermore, from an application perspective, this paper enables real-time adaptive control of reaction-diffusion PDEs which govern a series of real-world applications including, but not limited to, chemical reactions \cite{https://doi.org/10.1002/anie.200905513}, tubular reactor systems \cite{1024339}, multi-agent and social networking systems \cite{1643380, LEI20131326}, and Lithium-ion batteries \cite{7489035}.  
\paragraph{PDE backstepping for adaptive control}

    The first study of adaptive control for parabolic PDEs was in  \cite{4623267, SMYSHLYAEV20071543, SMYSHLYAEV20071557} which extended the adaptive backstepping results for nonlinear ODEs \cite{kkk} via three methodologies - the Lyapunov approach, passive identifier approach, and swapping approach. In this work, we focus on the Lyapunov approach which appears to exhibit superior transient performance as mentioned in \cite{4623267} and \cite{WANG2021109909}. This was then extended into hyperbolic PDEs which refer the reader to the rich set of literature \cite{8263680, ANFINSEN201786, ANFINSEN201772, ANFINSEN201869, ANFINSEN2018545}. In \cite{BRESCHPIETRI20092074} the authors then explored delay-adaptive control which was later extended in \cite{WANG2021109909} for unknown delays. For more complex systems, we refer the reader to adaptive backstepping schemes across a variety of challenging plants including coupled hyperbolic PDEs \cite{7963000}, coupled hyperbolic PDE-PDE-ODE cascades \cite{9656694} and the wave equation \cite{WANG2020108640}. Lastly, we mention the extensions into event-triggered adaptive control in \cite{WANG2021109637}, \cite{9735290}, \cite{KARAFYLLIS2019166}.
    
\paragraph{Kernel Implementations in PDE Control}

    The downside to the backstepping methodology is that, typically, the gain kernel is governed by a challenging infinite dimensional operator that needs to be approximated during implementation. Thus, we begin our discussion by reviewing a series of works on implementing backstepping controllers without learning. To start, we emphasize that backstepping is a late-lumping approach - i.e., one designs the control law in continuous space for the PDE and then discretizes during implementation \cite{AURIOL2019247, RIESMEIER202311407}. As such, \cite{GRUNE2022105237} exploits this property by decoupling the PDE system into finite dimensional slow and infinite dimensional fast subsystems for controller implementation. In a different approach, \cite{10384080, lin2024matlab} introduced a power series approximation of the gain kernel solution which is both simple and extremely useful for approximating complex kernels such as the Timoshenko beam, but scales poorly with respect to the spatial discretization size needed for accurately simulating PDEs. 

    Neural operators, in contrast, scale extremely well despite the high discretization needed for simulating PDEs \cite{lu2019deeponet}. They were first introduced in a series of work for approximating gain kernels in hyperbolic PDEs \cite{bhan_neural_2023} and parabolic PDEs \cite{krstic2023neural}, as well as general state estimation for nonlinear ODEs \cite{pmlr-v211-bhan23a}. This was then extended for hyperbolic and parabolic systems with delays in \cite{QI2024105714} and \cite{wang2023deep} respectively. Furthermore, \cite{wang2024backstepping} introduce operator approximations for coupled hyperbolic PDEs and \cite{pmlr-v211-yu24} applied neural operators for control of the Aw-Rascale-Zhang (ARZ) PDE with applications to traffic flows. Recently, \cite{vazquez2024gainonly} pointed out that in Hyperbolic PDEs, the approximation can be simplified to only $\hat{k}(1, y)$, but for parabolic PDEs, the entire kernel is needed in the control law. Lastly, we mention the first extension for approximating kernels that require recomputation online was in \cite{lamarque2024gain} for semi-linear hyperbolic PDEs via the gain scheduling approach and as aformentioned, this was later extended into adaptive control of hyperbolic PDEs in \cite{lamarque2024adaptive}.
\paragraph{Paper organization}
We first introduce, in Section \ref{sec:nominal}, the nominal adaptive backstepping control scheme for reaction-diffusion PDEs to expose the reader to the type of result we aim to maintain under the neural operator kernel approximation. We then explore and prove a series of properties for the gain kernel and its time derivative in Section \ref{sec:kernel-properties} in order to apply the universal approximation theorem \cite{lanthaler2023nonlocal} in Section \ref{sec:universal-approx}, to prove the existence of neural operators for approximating the gain kernels to arbitrary desired accuracy. We then state and prove our main result - namely a Lyapunov analysis of the full system with the NO approximated kernel, adaptive update law, and the backstepping control law in Section \ref{sec:main-result}. Lastly, in Section \ref{sec:simulations}, we conclude by presenting simulations demonstrating the efficacy of the proposed neural operator approximated control scheme on a parabolic PDE resembling molecular interactions in chemical reactions.  
\paragraph{Notation}
We use $\|\cdot\|_\infty$ for the infinity-norm, that is $ \|\lambda\|_\infty = \sup_{x \in [0, 1]} |\lambda(x)|$. Furthermore, we use $\|u(x, t)\|$ to be the spatial $L^2$ norm, $\|u(x, t)\| = \left( \int_0^1 u(x, t)^2 dx\right)^{\frac{1}{2}}$. We use $C^n(U; V)$ to indicate functions from set $U$ into set $V$ that have $n$ continuously differentiable derivatives. For scenarios, where the function has multiple arguments, ie, $f(x, y, t)$, we use $C^2_{x, y}C^1_t$ to indicate the function has continuous second derivatives in $x$ and $y$, but only continuous first derivatives with respect to $t$. If, the second argument of above is not given, say ie, $C^1(\mathbb{R}^p)$, then assume the function is mapping into the real numbers $\mathbb{R}$. Lastly, we denote the positive reals by $\rplus = \{x \in \mathbb{R} | x \geq 0 \}$ and Hilbert spaces by $H^n$. For example, $H^2$ is the space of functions with a $L^2$ weak derivative of order $2$. 
\section{Nominal controller for PDE backstepping} \label{sec:nominal}
We begin by introducing the following 1D Reaction-Diffusion PDE with a spatially varying coefficient $\lambda(x)$,
\begin{eqnarray}
    u_t(x, t) &=& u_{xx}(x, t) + \lambda(x)u(x, t)\,, \quad x \in (0, 1) \,,  \label{eq:parabolicMain1}\\
    u(0, t)&=& 0\,, \label{eq:parabolicMain2}\\ 
    u(1, t) &=& U(t) \label{eq:parabolicMain3}\,,
\end{eqnarray}
where $u(x, t)$ is defined for all $t \in \rplus$ with initial condition $u(x, 0) = u_0(x) \in H^2(0, 1)$ that is compatible with the boundary conditions. Further, $\lambda(x): [0, 1] \rightarrow \mathbb{R}$ is an unknown, spatially varying coefficient function that will be estimated online. The standard approach for controlling the PDE \eqref{eq:parabolicMain1}, \eqref{eq:parabolicMain2}, \eqref{eq:parabolicMain3} is to introduce the backstepping transformation
\begin{equation} \label{eq:backsteppingTransform}
    w(x, t) = u(x, t) - \int_0^x k(x, y) u(y, t) dy\,,
\end{equation}
to convert the system into the stable target system
\begin{eqnarray}
    w_t &=& w_{xx} \label{eq:perfectTargetSystem1}\,, \\
    w(0, t) &=& 0 \label{eq:perfectTargetSystem2}\,, \\ 
    w(1, t) &=& 0 \label{eq:perfectTargetSystem3}\,,
\end{eqnarray}
under the feedback control law
\begin{eqnarray}
    U(t) = \int_0^1 k(1, y) u(y, t) dy \label{eq:perfectFeedback}\,.
\end{eqnarray}
To ensure the transformation \eqref{eq:backsteppingTransform} converts \eqref{eq:parabolicMain1}, \eqref{eq:parabolicMain2}, \eqref{eq:parabolicMain3} into \eqref{eq:perfectTargetSystem1}, \eqref{eq:perfectTargetSystem2}, \eqref{eq:perfectTargetSystem3}, the kernel function $k$ must satisfy
\begin{eqnarray}
    k_{xx}(x, y) - k_{yy}(x, y) &=& \lambda(y) k(x, y), \quad (x, y) \in \breve{\mathcal{T}}\,, \label{eq:nominalGainKernel1} \\ 
    k(x, 0) &=& 0\label{eq:nominalGainKernel2} \,, \\
    k(x, x) &=& -\frac{1}{2} \int_0^x \lambda(y) dy \label{eq:nominalGainKernel3}\,,
\end{eqnarray}
where we define the triangular domains $\breve{\mathcal{T}} =\{0 < y \leq x < 1\}$ and $\mathcal{T}=\{0 \leq y \leq x \leq 1\}$. Note, the gain kernel, which is the solution to an infinite dimensional PDE that is not analytically solvable, is explicitly an operator mapping from functions of the spatially varying $\lambda(x)$ into the PDE solution function $k(x,y)$. However, in the adaptive control case, $\lambda(x)$ is unknown and thus needs to be estimated online via some approximation $\hat{\lambda}(x)$. Thus, the PDE kernel in \eqref{eq:nominalGainKernel1}, \eqref{eq:nominalGainKernel2}, \eqref{eq:nominalGainKernel3} becomes a mapping from $\hat{\lambda}(x) \mapsto k(x, y)$ where the solution \textit{requires recomputation at every timestep}. In what follows, we will denote the PDE solution for $\hat{\lambda}$ at time $t$ by $\breve{k}(x, y, t)$, and denote the neural operator approximation of the PDE solution for $\hat{\lambda}$ at time $t$ by $\hat{k}(x, y, t)$. Further, to ensure the mapping from $u$ into the $w$ system is well defined, recall the inverse backstepping transformation,
\begin{eqnarray} \label{eq:inverse-transform}
    u(x, t) = w(x, t) + \int_0^x l(x, y) u(y, t) dy\,,
\end{eqnarray}
where we refer to $l$ as the inverse backstepping kernel. 
Then, following Chapter 11 of \cite{smyshlyaevBook}, the estimator of $\hat{\lambda}$ is given by
\begin{eqnarray}
    \hat{\lambda}_t(x, t) &:=& \text{\rm Proj}(\phi(x, t), \hat{\lambda}(x, t))\label{eq:lambdaEstimator1}\,, \\ 
    \phi(x, t) &:=& \gamma \frac{u(x, t)}{1 + \|w\|^2} \nonumber \\ && \times \left(w(x, t) - \int_x^1 \breve{k}(y, x, t) w(y, t) dy\right) \label{eq:lambdaEstimator2}\,,
\end{eqnarray}
where $\gamma, \bar{\lambda} > 0$ are constants, $\|\lambda\|_\infty \leq \bar{\lambda}$, and the projection is defined as 
\begin{equation}
    \text{\rm Proj}(a, b) := \begin{cases}
        0, & \text{if } |b| = \bar{\lambda} \text{ and } ab > 0\\
        a\,. & \text{otherwise}
    \end{cases}
\end{equation}
 Thus, noting that we introduced a bound on $\lambda$, we formally state the only required assumption on $\lambda$ as the following.
\begin{assumption}
    $\lambda \in C^1([0, 1])$ and there exists a constant $\bar{\lambda} > 0$ such that $\|\lambda\|_\infty  \leq \bar{\lambda}$. 
\end{assumption}
Such an assumption is standard in the backstepping literature as $\lambda \in C^1([0, 1])$ is needed to ensure well-posedness of the kernel PDE \eqref{eq:nominalGainKernel1}, \eqref{eq:nominalGainKernel2}, \eqref{eq:nominalGainKernel3} and the bounded assumption is needed for the adaptive control estimate. Now, we state the main theorem for adaptive control of parabolic PDEs under the exact gains, which we aim to emulate under the neural operator approximated gains in this paper: under the feedback control scheme with the adaptive estimate $\hat{\lambda}$ and true kernel solution $\breve{k}$, the closed-loop system is regulated asymptotically to $0$.

\begin{theorem}
    (\cite{smyshlyaevBook} Stabilization under exact adaptive control scheme). There exists a $\gamma^*$ such that for $\gamma \in (0, \gamma^*)$, for any initial estimate $\hat{\lambda}(x, 0) \in C^1([0, 1])$ with $\|\hat{\lambda}(x, 0)\|_\infty \leq \bar{\lambda}$ and for any initial condition $u_0 \in H^2(0, 1)$ compatbile with boundary conditions, the classical solution of the closed loop system $(u, \hat{\lambda})$ consisting of the plant \eqref{eq:parabolicMain1}, \eqref{eq:parabolicMain2}, \eqref{eq:parabolicMain3}, the update law \eqref{eq:lambdaEstimator1}, \eqref{eq:lambdaEstimator2}, and the control law \eqref{eq:perfectFeedback} is bounded for all $(x, t) \in [0, 1] \times \rplus$ such that
    \begin{equation}
        \lim_{t\to \infty} \sup_{x \in [0, 1]} |u(x, t)| = 0. 
    \end{equation}
\end{theorem}
\section{Properties of gain kernel PDE} \label{sec:kernel-properties}
In order to approximate the gain kernel PDE for various $\hat{\lambda}$ estimates, we need to prove that both the time and spatial derivatives are well-defined, continuous and bounded. The exact use of these Lemmas will become clear in invoking the universal operator approximation theorem in Section \ref{sec:universal-approx} and the proof of our main result in Section \ref{sec:main-result}. We will employ the standard approach of successive approximations on the integral representation of the kernel PDE. Thus, recall the integral representation of the kernel PDE as
\begin{align}
    G(\xi, \eta, t) &=& -\frac{1}{4} \int_\eta^\xi \hat{\lambda}\left(\frac{s}{2}, t)\right) ds \nonumber \\ && + \frac{1}{4} \int_\eta^\xi \int_0^\eta \hat{\lambda}\left(\frac{\sigma - s}{2}, t \right) G(\sigma, s, t) ds d\sigma \label{eq:integralFormOfKernel} \,,
\end{align}
where 
\begin{eqnarray}
    \xi &=& x+y, \quad \eta = x-y\,, \label{eq:integralFormTransform1} \\ 
    G(\xi, \eta, t) &=& \breve{k}\left(\frac{\xi+\eta}{2}, \frac{\xi-\eta}{2}, t \right), \quad (x, y) \in \mathcal{T}\,. \label{eq:integralFormTransform2}
\end{eqnarray}
Then, for the reader's convenience, we briefly recall the following lemma whose proof is well known in the backstepping literature. 

\begin{lemma} (\cite{1369395} Existence and bound for gain kernel.)
    Let $\bar{\lambda} > 0$ such that $\|\hat{\lambda}\|_\infty \leq \bar{\lambda}$ for all $(x, t) \in [0,1] \times \rplus$. Then, for any fixed $t \in \rplus$, and for every $\hat{\lambda}(x, t) \in C^1([0, 1])$ at fixed $t$, the kernel governed by the PDE \eqref{eq:nominalGainKernel1}, \eqref{eq:nominalGainKernel2}, \eqref{eq:nominalGainKernel3} with function $\hat{\lambda}$ has a unique $C^2(\mathcal{T})$ solution with the bound
    \begin{equation}
        \left\|\breve{k}(x, y, t)\right\|_\infty \leq \bar{\lambda}e^{2\bar{\lambda}x}. \label{eq:gainKernelBnd}
    \end{equation}
\end{lemma} 
\begin{lemma}(Existence and bound for $\breve{k}_x(x, x, t)$)
    Let $\bar{\lambda} > 0$ such that $\|\hat{\lambda}\|_\infty \leq \bar{\lambda}$, $\forall (x, t) \in [0, 1] \times \rplus$. Then, for any fixed $t \in \rplus$, $\breve{k}_x(x,x,t) \in C^1[0, 1]$ with the bound
    \begin{equation}
        \left \|\breve{k}_x(x, x, t)\right \|_\infty \leq \frac{1}{2}\bar{\lambda}. \label{eq:gainKernelSpatialDerivBnd}
    \end{equation}
\end{lemma}
\begin{proof}
    One can show the existence and continuity of this derivative on all of $(x, y) \in \mathcal{T}$ by differentiating \eqref{eq:integralFormOfKernel} and using the method of successive approximations. However, we only require a bound at the boundary condition $y=x$ and thus, one has
    \begin{eqnarray}
        \left|\frac{\partial}{\partial x}\breve{k}(x, x, t)\right| = \left|-\color{black}\frac{1}{2} \hat{\lambda}(x, t) \color{black} \right| \leq \frac{1}{2} \bar{\lambda}, \\ \quad \forall (x, t) \in [0,1] \times \rplus \nonumber\,.
    \end{eqnarray}
\end{proof}
\begin{lemma} \label{lem:ktbound} (Existence and bound for $\breve{k}_t(x, y, t)$) Let $\bar{\lambda} > 0$ such that $\|\hat{\lambda}\|_\infty \leq \bar{\lambda}$, $\forall (x, t) \in [0, 1] \times \rplus$. Then, for every $\hat{\lambda}(x, t) \in C^1([0, 1] \times \rplus)$, $\breve{k}_t(x, y, t)$ has a unique $C^0(\mathcal{T} \times \rplus)$ solution with the bound
\begin{eqnarray}
    \left \|\breve{k}_t(t)\right\| &\leq& M\|\hat{\lambda}_t(t)\|\,, \quad \forall t \geq 0\,, \label{eq:kt-bound-lem}\\
    M &=& e^{2\bar{\lambda}}(1+\bar{\lambda}e^{2\bar{\lambda}})\,. \label{eq:def-M}
\end{eqnarray}
\end{lemma}
\begin{proof}
    We begin by showing existence and continuity. Differentiating \eqref{eq:integralFormOfKernel} with respect to $t$ yields
    \begin{eqnarray}
        G_t(\xi, \eta, t) &=& -\frac{1}{4}\int_\eta^\xi \hat{\lambda}_t\left(\frac{s}{2}, t \right) ds \nonumber \\ && + \frac{1}{4} \int_\eta^\xi \int_0^\eta \bigg[ \hat{\lambda}_t\left(\frac{\sigma - s}{2}, t \right) G(\sigma, s, t) \nonumber \\ && + \hat{\lambda}\left(\frac{\sigma - s}{2}, t \right) G_t(\sigma, s, t) \bigg] ds d\sigma\,. \label{eq:integralKernelDerivWRTTime}
    \end{eqnarray}
    Define the iterate sequence 
    \begin{alignat}{3}
        G_t^0(\xi, \eta, t) &:=&-&\frac{1}{4} \int_\eta^\xi \hat{\lambda}_t\left(\frac{s}{2}, t \right) ds  \nonumber \\ 
        &&+&\frac{1}{4} \int_\eta^\xi \int_0^\eta \hat{\lambda}_t\left(\frac{\sigma - s}{2}, t \right) G(\sigma, s, t)ds d\sigma \,, \\ 
        G_t^{n+1}(\xi, \eta, t) &:= &&\frac{1}{4} \int_\eta^\xi \int_0^\eta \hat{\lambda}\left(\frac{\sigma - s}{2}, t \right) G_t^n(\sigma, s, t) ds d \sigma\,,
    \end{alignat}
    and consider the difference sequence 
    \begin{alignat}{3}
        \Delta G_t^0 = G_t^0 &=&-&\frac{1}{4} \int_\eta^\xi \hat{\lambda}_t\left(\frac{s}{2}, t \right) ds  \nonumber \\ 
        &&+&\frac{1}{4} \int_\eta^\xi \int_0^\eta \hat{\lambda}_t\left(\frac{\sigma - s}{2}, t \right) G(\sigma, s, t)ds d\sigma \,, \label{eq:deltag0} \\ 
        \Delta G_t^{n+1} = G_t^{n+1} - G_t^n &=& &\frac{1}{4} \int_\eta^\xi \int_0^\eta \hat{\lambda}\left(\frac{\sigma - s}{2}, t \right) \Delta G_t^n(\sigma, s, t) ds d \sigma\,. \label{eq:deltagn}
    \end{alignat}
    Then it is clear $G_t^n = \sum_{n=0}^\infty \Delta G_t^n$. We aim to show this series uniformly converges. For notational simplicity, we introduce the coefficient $\alpha(T)$ which specifies the max of $\hat{\lambda}$ and its derivative up to a time $T$.
    \begin{equation}
        \alpha(T) := \max \left\{\sup_{t \in [0, T]} \left\|\hat{\lambda}_t(\cdot, t)\right\|, \hat{\lambda}\right\}\,. \label{eq:coefficientFunc}
    \end{equation}
    Then, for any $T \geq t$, the difference sequence satisfies the following bounds
    \begin{eqnarray}
        \left|\Delta G_t^0\right| &\leq& \alpha(T) \left(\frac{1}{2} + \alpha(T) e^{2\alpha(T)}\right) \,, \label{eq:delta0bnd} \\ 
        \left|\Delta G_t^{n}\right|  &\leq& \left(\frac{1}{2} + \alpha(T)e^{2\alpha(T)}\right) \frac{\alpha(T)^{n+1}(\eta+\xi)^n}{n!} \label{eq:deltagbnd}\,.
    \end{eqnarray}
    The first bound comes from applying \eqref{eq:coefficientFunc} to \eqref{eq:deltag0}. The second bound can be shown via induction.  Assume \eqref{eq:deltagbnd} holds for $G_t^n$. Then, substituting into \eqref{eq:deltagn} yields
    \begin{eqnarray}
        \left| \Delta G_t^{n+1}\right| &\leq& \frac{1}{4} \int_\eta^\xi \int_0^\eta \hat{\lambda} \left(\frac{\sigma - s}{2}, t \right) \nonumber \\&& \times \bigg(\left(\frac{1}{2} + \alpha(T)e^{2\alpha(T)}\right)  \nonumber
        \\ && \times \frac{\alpha(T)^{n+1}(\eta+\xi)^n}{n!} \bigg) ds d\sigma \nonumber  \\ 
        &\leq& \frac{1}{4} \left(\left(\frac{1}{2} + \alpha(T)e^{2\alpha(T)} \right) \frac{\alpha(T)^{n+2}}{n!} \right) \nonumber \\ && \times \int_\eta^\xi \int_0^\eta (\sigma + s)^{n+1} ds d\sigma \nonumber \\
        &\leq&  \frac{1}{4} \bigg(\left(\frac{1}{2} + \alpha(T)e^{2\alpha(T)} \right) \nonumber \\ && \times  \frac{\alpha(T)^{n+2}(\eta + \xi)^{n+1}}{(n+1)!} \bigg)\,.
    \end{eqnarray}
    Thus, the series $\sum_{n=0}^\infty \Delta G_t^n$ uniformly converges on $\mathcal{T} \times [0, T]$ for $T \geq t$. Since there is a one-to-one correspondence between $G$ and $\breve{k}$, this implies the existence and continuity of $\breve{k}_t$ on $\mathcal{T} \times \mathbb{R}_+$. One could find a bound for $\breve{k}_t$ by analyzing the convergence bound on the series in \eqref{eq:delta0bnd}, \eqref{eq:deltagbnd}, but for the future Lyapunov analysis, it is easier to work with a bound on $\breve{k}$ in terms of the estimate $\hat{\lambda}$. 

    From \eqref{eq:integralKernelDerivWRTTime}, the triangle inequality, Cauchy Schwarz, and the substitution of \eqref{eq:gainKernelBnd} we obtain 
    \begin{eqnarray}
        \left|G_t(\xi, \eta, t)\right| &=& \Bigg|-\frac{1}{4}\int_\eta^\xi \hat{\lambda}_t\left(\frac{s}{2}, t \right) ds \nonumber \\ && + \frac{1}{4} \int_\eta^\xi \int_0^\eta \bigg[ \hat{\lambda}_t\left(\frac{\sigma - s}{2}, t \right) G(\sigma, s, t) \nonumber \\ && + \hat{\lambda}\left(\frac{\sigma - s}{2}, t \right) G_t(\sigma, s, t) \bigg] ds d\sigma\,\bigg| \\
        &\leq& \left\|\hat{\lambda}_t\right\| + \left\|\hat{\lambda}_t\right\| \bar{\lambda} e^{2 \bar{\lambda}} \nonumber \\ &&+ \bar{\lambda} \int_\eta^\xi \int_0^\eta |G_t(\sigma, s, t)| ds d\sigma \,.
    \end{eqnarray}
    Now, applying Fubini's theorem and noting the integrand is always non-negative yields
    \begin{eqnarray}
        \left|G_t(\xi, \eta, t)\right|  &=&  \left\|\hat{\lambda}_t\right\| + \left\|\hat{\lambda}_t\right\| \bar{\lambda} e^{2 \bar{\lambda}} \nonumber \\ &&+ \bar{\lambda} \int_\eta^\xi \int_0^\eta |G_t(\sigma, s, t)| ds d\sigma  \\ 
        &\leq& \left\|\hat{\lambda}_t\right\| + \left\|\hat{\lambda}_t\right\|\bar{\lambda} e^{2 \bar{\lambda}} \nonumber \\ &&+ \bar{\lambda} \int_0^\eta \int_0^\xi |G_t(\sigma, s, t)| ds d\sigma\,.    
    \end{eqnarray}
    Now we apply the Wendroff inequality (see \ref{appendix:wendroff}) first given in the book by Beckenbach and Bellman \cite{Beckenbach_Bellman_1961} as an extension to Gronwall's inequality in 2D to obtain the following result for $G_t$
    \begin{eqnarray} \label{eq:kt-bound}
        |G_t(\xi, \eta, t)| \leq \left\|\hat{\lambda}_t\right\| (1+\bar{\lambda}e^{2\bar{\lambda}}) e^{2\bar{\lambda}}\,. 
    \end{eqnarray}
    The final result then holds for $\breve{k}_t$ given $G_t$ is an equivalent integral representation of the time derivative of the kernel PDE. 
\end{proof}
\section{Neural operator approximation of the gain kernel} \label{sec:universal-approx}
We aim to approximate the kernel mapping $\hat{\lambda} \mapsto k$ by a neural operator and thus begin by presenting a general universal approximation theorem for the nonlocal neural operator - a unifying framework that encompasses a series of operator learning architectures including both the popular FNO \cite{li2021fourier} and DeepONet \cite{Lu2021} frameworks and their extensions such as NOMAD \cite{seidman2022nomad} and the Laplace NO \cite{cao2023lno}. We give the details of the nonlocal-neural operator architecture in \ref{appendix:nonlocal} (as well as its connections to FNO and DeepONet) and refer the reader to \cite{lanthaler2023nonlocal} for further details. 

\begin{theorem} \label{thm:nnoUniversalApprox} (\cite[Theorem 2.1]{lanthaler2023nonlocal} Neural operator approximation theorem.) Let $\Omega_u \subset \mathbb{R}^{d_{u_1}}$ and $\Omega_v \subset \mathbb{R}^{d_{v_1}}$ be two bounded domains with Lipschitz boundary. Let $\mathcal{G}: C^0(\overline{\Omega_u};\mathbb{R}^{d_{u_2}}) \rightarrow C^0(\overline{\Omega_v}; \mathbb{R}^{d_{v_2}}) $ be a continuous operator and fix a compact set $K \subset C^0(\overline{\Omega_u};\mathbb{R}^{d_{u_2}})$. Then for any $\epsilon > 0$, there exists a nonlocal neural operator $\hat{\mathcal{G}}: K \rightarrow C^0(\overline{\Omega_v}; \mathbb{R}^{d_{v_2}})$ such that 
\begin{equation}
    \sup_{u \in K} |\mathcal{G}(u)(y) - \hat{\mathcal{G}}(u)(y)|\leq \epsilon\,,
\end{equation}
for all values $y \in \Omega_v$.
\end{theorem}
For readers familiar with the previous explorations of neural operators in approximating kernel gain functions \cite{bhan_neural_2023}, we provide the following corollary for DeepONet. 
\begin{corollary} (DeepONet universal approximation theorem; first proven in \cite{lu2021advectionDeepONet})
Consider the setting of Theorem \ref{thm:nnoUniversalApprox}. Then, for all $\epsilon > 0$, there exists $p^*$, $m^*$ such that for all $p \geq p^*$, $m\geq m^*$, there exists neural network weights $\varphi^{(k)}$, $\theta^{(k)}$ such that the neural networks (see \cite{bhan_neural_2023}, section III for definition) $g^\mathcal{N}$ and $f^\mathcal{N}$ in the DeepONet given by 
\begin{eqnarray}
    G_\mathbb{N}(\bm{u}_m)(y) = \sum_{k=1}^p g^\mathcal{N}(\bm{u}; \varphi^{(k)})f^\mathcal{N}(y; \theta^{(k)})\,,
\end{eqnarray}
satisfy 
\begin{eqnarray}
    \sup_{u \in K} |\mathcal{G}(u)(y) - \mathcal{G}_\mathbb{N}(u)(y)|\leq \epsilon\,,
\end{eqnarray}
for all values $y \in \Omega_v$.
\end{corollary}
Note that Theorem \ref{thm:nnoUniversalApprox} has two main assumptions. First, the input function space is required to be compact. Second, the operator mapping to be approximated must be continuous. We now introduce the operator for $\breve{k}$ that we aim to approximate, noting that the operator output includes $\breve{k}$ as well as its derivatives $\breve{k}_x, \breve{k}_{xx}, \breve{k}_t$ which is needed for proving stability under the neural operator approximation. 

Define the set of functions
\begin{equation}
    \underline{K} = \{k \in C^2_{x, y}C_t^1(\mathcal{T} \times \rplus)|k(x, 0, t) = 0\}, \quad \forall x \in [0, 1], t  \in \rplus\}\,.
\end{equation} 
Let $\Lambda$ be a compact set of $C^0([0, 1])$ with the supremum norm such that for every $\lambda \in \Lambda$, $\|\lambda\|_\infty < M$ and $\lambda$ is $R-$Lipschitz where $M, R > 0$ are constants that can be as large as needed. Then, denote the operator $\mathcal{K}: \Lambda \rightarrow \underline{K}$ as 
\begin{equation}
    \mathcal{K}(\hat{\lambda}(\cdot, t)) := \breve{k}(x, y, t)\,.
\end{equation}
Further, define the operator $\mathcal{M}: \Lambda \rightarrow \underline{K} \times C^1([0, 1] \times \rplus) \times C^0_{x, y} C^1_t (\mathcal{T} \times \rplus)$ such that 
\begin{equation}
    \mathcal{M}(\hat{\lambda}(\cdot, t)) := (\breve{k}(x, y, t), \kappa_1(x, t), \kappa_2(x, y, t))\,,
\end{equation}
where
\begin{eqnarray}
    \kappa_1(x, t) &=& 2 \frac{\partial}{\partial x} (\breve{k}(x, x, t)) + \hat{\lambda}(x, t)\,, \\ 
    \kappa_2(x, y, t) &=& \breve{k}_{xx}(x, y, t) - \breve{k}_{yy}(x, y, t) \nonumber \\ && - \hat{\lambda}(y, t) \breve{k}(x, y, t)\,.
\end{eqnarray}
For the Lyapunov analysis that follows, we also need to include $\breve{k}_t$ in our approximation. Thus, define the operator $\mathcal{K}_1: \Lambda^2 \rightarrow C^2_{x, y}C^0_t(\mathcal{T} \times \rplus)$ such that
\begin{equation}
    \mathcal{K}_1(\hat{\lambda}(\cdot, t), \hat{\lambda}_t(\cdot, t)) := \breve{k}_t(x, y, t)\,,
\end{equation}
where under the transformations $\xi = x+y, \eta=x-y,  \forall(x, y) \in \mathcal{T}$, we have $\breve{k}_t$  is the solution to the integral equation  
\begin{alignat}{3}
    \breve{k}_t\left(\frac{\xi+\eta}{2}, \frac{\xi-\eta}{2}\right) &=&& G_t(\xi, \eta, t)\,, \\ 
    0 &=& -& G_t(\xi, \eta, t) -\frac{1}{4}\int_\eta^\xi \hat{\lambda}_t\left(\frac{s}{2}, t \right) ds  \nonumber \\ 
    &&+& \frac{1}{4} \int_\eta^\xi \int_0^\eta \bigg[ \hat{\lambda}_t \left(\frac{\sigma - s}{2}, t \right) G(\sigma, s, t) \nonumber \\ &&+&\hat{\lambda}\left(\frac{\sigma - s}{2}, t \right) G_t(\sigma, s, t) \bigg] ds d\sigma\,, 
\end{alignat}
where $G$ is given by \eqref{eq:integralFormOfKernel}. 

Lastly, consider the composition of the operators $\mathcal{M}$ and $\mathcal{K}_1$ given as $\mathcal{N}: \Lambda^2 \rightarrow \underline{K} \times C^1([0, 1] \times \rplus) \times C^0_{x, y}C^1_t(\mathcal{T} \times \rplus) \times C^2_{x, y}C_t^0(\mathcal{T} \times \rplus)$ such that 
\begin{equation}
    \mathcal{N}(\hat{\lambda}(\cdot, t), \hat{\lambda}_t(\cdot, t)) := (\mathcal{M}(\hat{\lambda}(\cdot, t)), \mathcal{K}_1(\hat{\lambda}(\cdot, t), \hat{\lambda}_t(\cdot, t)))\,.
\end{equation}
We now aim to approximate the operator $\mathcal{N}$ which requires continuity of the operator. First, we note that $\mathcal{M}$ was shown to be continuous in \cite[Theorem 4]{krstic2023neural}. Thus, it suffices to show $\mathcal{K}_1$ is continuous.
\begin{lemma}
    Fix $t\geq 0$. Let $\lambda_1(\cdot, t)$, $\lambda_2(\cdot, t) \in \Lambda$ and $k_1 = \mathcal{K}(\lambda_1), k_2=\mathcal{K}(\lambda_2)$. Then, $\mathcal{K}_1$ is Lipschitz continuous. Explicitly, there exists a Lipschitz constant $A > 0$ such that 
    \begin{equation}
        \left\|\mathcal{K}_1\left(\lambda_1, \frac{\partial}{\partial t} \lambda_1 \right) - \mathcal{K}_1\left(\lambda_2, \frac{\partial}{\partial t} \lambda_2 \right)\right\|_\infty \leq A \|\lambda_1 - \lambda_2\|_\infty. 
    \end{equation}
\end{lemma}
\begin{proof}
Let $\lambda_1, \lambda_2 \in \Omega$. Let $k_1 = \mathcal{K}(\lambda_1)$, $k_2 = \mathcal{K}(\lambda_2)$ and denote $G_1, G_2$ the corresponding transforms of $k_1, k_2$ in the integral form $G_i = k_i\left(\frac{\xi+\eta}{2}, \frac{\xi - \eta}{2} \right)$, $i \in \{1, 2\}$. Then, we have the following
\begin{alignat}{3}
    \delta \lambda &=&& \lambda_1 - \lambda_2\,, \\ 
    \delta G &=&& G_1 - G_2\,, \\ 
    \delta G_t &= &-& \nonumber  \color{black} \frac{1}{4} \int_\eta^\xi \frac{ \partial (\delta \lambda)}{\partial t} \left(\frac{s}{2}, t \right) ds \\  \color{black}
    &&+& \color{black}\frac{1}{4} \int_\eta^\xi \int_0^\eta \bigg [ \frac{\partial (\delta \lambda)}{\partial t}\left(\frac{\sigma - s}{2}, t \right) G_1(\sigma, s, t) \nonumber \\  &&+& \color{black}\frac{\partial \lambda_2 \left(\frac{\sigma -s}{2} , t \right)}{\partial t} \delta G(\sigma, s, t) \nonumber \\ &&+& \color{black}\delta \lambda \left(\frac{\sigma-2}{2}\right) \frac{ \partial G_1}{\partial t}(\sigma, s, t) \nonumber \\ &&+& \color{black}\lambda_2\left(\frac{\sigma-2}{2}\right) \frac{\partial (\delta G)}{\partial t}(\sigma, s, t) \bigg] ds d\sigma \color{black}
\end{alignat}
In the proof that follows, will omit the arguments on $\lambda$ and $G$ for conciseness. Define the sequence
\begin{alignat}{3}
    \delta G_t^{n+1} &=& &\frac{1}{4} \int_\eta^\xi \int_0^\eta \lambda_2 \frac{\partial (\delta G)}{\partial t}^{\color{black}n\color{black}} ds d\sigma\,, \\ 
    \delta G_t^0 &=& -& \frac{1}{4} \int_\eta^\xi \frac{\partial (\delta \lambda)}{\partial t} ds \nonumber \\ &&+& \frac{1}{4} \int_\eta^\xi \int_0^\eta \bigg[ \color{black}\frac{\partial (\delta \lambda)}{\partial t} \color{black} G_1 + \color{black} \frac{\partial \lambda_2}{\partial t} \color{black}\delta G + \delta \lambda \frac{\partial G_1}{\partial t}\bigg] ds d\sigma\,.
\end{alignat}
Then, using \cite{krstic2023neural} the Lipschitz bound for $\delta G$, the fact $\lambda, \lambda_t \in \Lambda$, and the bound in Lemma \ref{lem:ktbound}, for any $T \geq t$, we attain 
\begin{eqnarray}
    \|\delta G_t^0 \|_\infty  &\leq& A \|\delta \lambda\|_\Lambda \,,\\ 
    |\delta G_t^n | &\leq& A \|\delta \lambda\|_\Lambda \frac{\alpha(T)^n \color{black}{(\xi - \eta)} \color{black}^n}{n!}\,, \\ 
    A(\alpha(T)) &>& 0\,,
\end{eqnarray}
\color{black}
where we introduce the norm $\|\lambda \|_\Lambda := \|\lambda \|_\infty + \|\lambda_t \|_\infty $. \color{black}
The result then follows from the uniform convergence of the series
\begin{eqnarray}
    \delta G_t &=& \sum_{n=0}^\infty \delta G_t^n\,, \\
    \|\delta G_t\|_\infty &=& A \|\delta \lambda\|_\Lambda e^{\alpha(T)}\,.
\end{eqnarray}
\end{proof}

Using the continuity of $\mathcal{N}$, we now invoke Theorem \ref{thm:nnoUniversalApprox} to approximate the composed kernel mapping.
\begin{theorem} \label{thm:operatorApproxOfKernel}
    Fix $t\geq 0$ and let $(\hat{\lambda}(\cdot, t), \hat{\lambda}_t(\cdot, t)) \in \Lambda^2$. Then for all $\epsilon > 0$, there exists a neural operator $\hat{\mathcal{N}}$ such that for all $(x, y) \in \mathcal{T}$
    \begin{alignat}{2}
        &|\mathcal{K}(\hat{\lambda}(\cdot, t) - \hat{\mathcal{K}}(\hat{\lambda})(\cdot, t)| \nonumber &\\
        &+|2 \partial_x(\mathcal{K}(\hat{\lambda})(x,x,t) - \hat{\mathcal{K}}(\hat{\lambda}))| \nonumber &\\ 
        &+|(\partial_{xx} - \partial_{yy})(\mathcal{K}(\hat{\lambda})(\cdot, t) - \hat{\mathcal{K}}(\hat{\lambda})(\cdot, t))\nonumber  \\ \nonumber  &- \hat{\lambda}(y)(\mathcal{K}(\hat{\lambda})(\cdot, t) - \hat{\mathcal{K}}(\hat{\lambda})(\cdot, t))| \nonumber & \\ 
        &+ |\mathcal{K}_1(\hat{\lambda}, \hat{\lambda}_t)(\cdot, t) - \hat{\mathcal{K}}_t(\hat{\lambda}, \hat{\lambda}_t)(\cdot, t)| &\leq \epsilon \,.
    \end{alignat}
\end{theorem}
\section{Stability under NO Approximated Gain Kernel} \label{sec:main-result}
   To simplify notation, define the following constants for the bounds for both the approximate backstepping kernel and approximate inverse backstepping kernel as 
    \begin{eqnarray}
        \|\hat{k}\|_\infty = \breve{k} + \tilde{k} \leq \bar{\lambda}e^{2\bar{\lambda}}+\epsilon &=:& \bar{k}\,, \\ 
        \|\hat{l}\|_\infty \leq \|\hat{k}\|_\infty e^{\|\hat{k}\|_\infty} \leq \bar{k}e^{\bar{k}} &=:& \bar{l}.
    \end{eqnarray}
We now present our main result.

\begin{theorem} \label{thm:main-result}
    Let $\bar{\lambda} > 0$. For any Lipschitz $\lambda, \hat{\lambda}(\cdot, 0) \in C^1([0, 1])$ such that $\|\lambda\|_\infty, \|\hat{\lambda}(\cdot, 0)\|_\infty \leq  \bar{\lambda}$ and for all neural operator approximations $\hat{k} = \mathcal{K}(\hat{\lambda})$ with accuracy $\epsilon \in (0, \epsilon^*)$ from Theorem \ref{thm:operatorApproxOfKernel}
    where $\epsilon^*$ is the unique solution to the equation
    \begin{equation}
        \epsilon^*\left(1+(\epsilon^*+\bar{\lambda}e^{2\bar{\lambda}})e^{\epsilon^* + \bar{\lambda}e^{2\bar{\lambda}}} \right) = 1/12\,        
    \end{equation}
    (whose left side being zero at zero and monotonic implies the existence and uniqueness of $\epsilon^*$), 
    for all $\gamma \in (0, \gamma^*(\epsilon, \bar{\lambda}))$, where
    \begin{equation} \label{eq:main-gammastar}
        \gamma^*(\epsilon, \bar{\lambda}) = \frac{\frac{1}{4}-3(1+\bar{l})\epsilon}{(1+\bar{l})^2(1+\bar{k})^2} > 0 \,,
    \end{equation}
    and for all initial estimates $\hat{\lambda}(\cdot, 0), \hat{\lambda}_t(\cdot, 0) \in \Lambda$ and all initial condition $u_0 \in H^2(0, 1)$ compatible with boundary conditions, the classical solution, for which $\hat{\lambda}_t(\cdot, t)$ remains in $\Lambda$ and $\hat{k}(\cdot, t)$ remains  differentiable, of the closed-loop system $(u, \hat{\lambda}, \mathcal{\hat{K}}(\hat{\lambda}))$ consisting of the plant \eqref{eq:parabolicMain1}, \eqref{eq:parabolicMain2}, \eqref{eq:parabolicMain3}, the update law 
    \begin{eqnarray} 
    \hat{\lambda}_t(x, t) &=& \text{\rm Proj}(\phi(x, t), \hat{\lambda}(x, t))\label{eq:mainResultUpdateLaw1}\,, \\ 
    \phi(x, t) &=& \gamma \frac{u(x, t)}{1 + \|\hat{w}\|^2} \nonumber \\ && \times \left(\hat{w}(x, t) - \int_x^1 \hat{k}(y, x, t) \hat{w}(y, t) dy\right) \label{eq:mainResultUpdateLaw2}\,, \\ 
    \hat{w}(x, t) &=& u(x, t) - \int_0^x \hat{k}(x, y, t)u(y, t) dy \,, \label{eq:mainResultUpdateLaw3}
    \end{eqnarray}
    and the controller 
    \begin{equation} \label{eq:finalController}
        U(t) = \int_0^1 \hat{k}(1, y, t)u(y, t) dy \,,
    \end{equation}
    is bounded for all $x \in [0, 1], t \in \rplus$ and 
    \begin{equation}
        \lim_{t \to \infty} \max_{x \in [0, 1]} |u(x, t)| = 0\,.
    \end{equation}
    Additionally, there exist constants $\rho, R > 0$ such that the  stability estimate 
    \begin{eqnarray}
        \Gamma(t) &\leq& R(e^{\rho \Gamma(0)} - 1)\,,  \label{eq:Gamma-stability-estimate} \\ 
        \Gamma(t) &:=& \int_0^1 \left[ u^2(x, t) + \left(\lambda(x) - \hat{\lambda}(x, t)\right)^2\right] dx\,, \label{eq:Gamma-def}
    \end{eqnarray}
    holds for all $t \geq 0$. 
\end{theorem}
\begin{proof}
    Consider the approximate backstepping transform
    \begin{eqnarray} \label{Eq:backsteppingkhat}
        \hat{w}(x, t) &=& u(x, t) - \int_0^x \hat{k}(x, y, t) u(y, t) dy\,,  \\ 
        u(x, t) &=& \hat{w}(x, t) + \int_0^x \hat{l}(x, y, t) \hat{w}(y, t)dy \,,\label{eq:backsteppinglhat}
    \end{eqnarray}
    where $\hat{k} = \mathcal{K}(\hat{\lambda})$ and $\hat{l}$ is the corresponding inverse backstepping transformation satisfying, 
    \begin{eqnarray}
        \hat{l}(x, y, t) = \hat{k}(x, y, t) + \int_{y}^x \hat{k}(x, \xi, t)\hat{l}(\xi, y, t) d\xi \,. 
    \end{eqnarray}
    Then, following \ref{appendix:target-system}, the target system becomes
    \begin{eqnarray} \label{eq:target}
        \hat{w}_t &=& \hat{w}_{xx} + \delta_{k0}(x, t) u(x, t) - \int_0^x \delta_{k1}(x, y, t)u(y, t) dy \nonumber \\ && - \int_0^x \tilde{\lambda}(y, t) \hat{k}(x, y, t) u(y, t) dy \nonumber \\&& - \int_0^x \hat{k}_t(x, y, t) u(y, t)\,, \\ 
        \hat{w}(0, t) &=& 0 \,,\\ 
        \hat{w}(1, t) &=& 0\,, 
    \end{eqnarray}
    where 
    \begin{align}
        \delta_{k0}(x, t) =&\tilde{\lambda}(x, t) - 2\tilde{k}_x(x, x, t) \,,\\ 
        \delta_{k1}(x, y, t)= & \tilde{k}_{xx}(x, y, t) - \tilde{k}_{yy}(x, y, t) - \hat{\lambda}(y)\tilde{k}(x, y, t)\,.
    \end{align}
    Note, such a target system has two perturbation terms given by $\delta_{k0}$ and $\delta_{k1}$ from the neural operator approximation as in \cite{krstic2023neural} as well as two perturbations from the adaptive control scheme - namely the parameter estimation error $\tilde{\lambda}$ and the rate of the parameter estimation gain $\hat{k}_t$.  
    Now, for constant $\gamma > 0$, consider the Lyapunov function
    \begin{equation}
        V = \frac{1}{2}\ln(1+\|\hat{w}\|^2) + \frac{1}{2\gamma}\|\tilde{\lambda}\|^2\,. \label{Eq:lyapunovFunc}
    \end{equation}
    Computing the time derivative along the system trajectories, applying Leibniz rule, substituting for $\hat{w}_t$ and noting $\frac{\partial}{\partial t} \tilde{\lambda} = \hat{\lambda}_t$ yields
    \begin{eqnarray}
        \dot{V} &=& \frac{1}{1|\|\hat{w}\|^2} \int_0^1 \hat{w}(x, t) \hat{w}_t(x, t) dx  \nonumber \\ && +\frac{1}{\gamma} \int_0^1 \tilde{\lambda}(x, t) \tilde{\lambda}_t(x, t) dx \\ 
       &=& \frac{1}{1|\|\hat{w}\|^2} \bigg(\int_0^1 \hat{w}(x, t) \hat{w}_{xx} dx  \nonumber \\ &&+ \int_0^1 \hat{w}(x, t) \delta_{k0}(x, t) u(x, t) dx \nonumber \\ &&+ \int_0^1 \hat{w}(x, t) \int_0^x \delta_{k1}(x, y, t) u(y, t) dy dx \nonumber \\ && - \int_0^1 \hat{w}(x, t) \int_0^x \tilde{\lambda}(y, t) \hat{k}(x, y, t)u(y, t) dy dx \nonumber \\ && - \int_0^1 \hat{w}(x, t)\int_0^x \hat{k}_t(x, y, t) u (y, t) dy dx \bigg)  \nonumber \\ && +\frac{1}{\gamma} \int_0^1 \tilde{\lambda}(y, t) \hat{\lambda}_t(y, t) dy\,.
    \end{eqnarray}
    Noting that $\hat{w}(1, t) = \hat{w}(0, t) = 0$ and applying integration by parts to the first term along with substituting the update law \eqref{eq:mainResultUpdateLaw1}, \eqref{eq:mainResultUpdateLaw2} computed with $\hat{k}$ in for $\hat{\lambda}_t$ yields
    \begin{eqnarray}
        \dot{V} &=& -\frac{1}{1 + \|\hat{w}\|^2} (I_1(x, t) + I_2(x, t) + I_3(x, t) \nonumber \\ &&  + I_4(x, t)) \,,\\ 
        I_1(x, t) &=& \int_0^1 \hat{w}_x^2(x, t) dx \,,\\ 
        I_2(x, t) &=& 2 \int_0^1 \hat{w}(x, t) \tilde{k}_x(x, x, t)u(x, t) dx \,,\\ 
        I_3(x, t) &=& \int_0^1 \hat{w}(x, t) \int_0^x \delta_{k1}(x, y, t) u(y, t) dy dx\,, \\ 
        I_4(x, t) &=& \int_0^1 \hat{w}(x, t) \int_0^x \hat{k}_t(x, y, t)u(y, t) dy dx \,. 
    \end{eqnarray}

    We now bound each term in terms of $\|\hat{w}_x\|^2$ to estimate a bound on $\dot{V}$. The $I_1$ term is obvious. 

    By expanding $u(x, t)$, using \eqref{eq:backsteppinglhat}, and substituting $\tilde{k}_x(x,x,t) < \epsilon/2$ from \ref{thm:operatorApproxOfKernel} yields
    \begin{eqnarray}
        |I_2(x, t)| &\leq& \epsilon \int_0^1 \hat{w}(x, t) \bigg(\hat{w}(x, t)\nonumber \\ && + \int_0^x \hat{l}(x, y, t) \hat{w}(y, t) dy \bigg)dx \\ 
        &\leq& \epsilon (1+\bar{l})\|\hat{w}\|^2\,,
    \end{eqnarray}
    and employing Poincar\'{e}'s inequality yields
    \begin{equation}
        |I_2(x, t)| \leq 4 \epsilon (1+\bar{l})\|\hat{w}_x\|^2\,.
    \end{equation}
    The same exact procedure is used to bound $I_3$ yielding
    \begin{equation}
        |I_3(x, t)| \leq 4 \epsilon (1+\bar{l}) \|\hat{w}_x\|^2\,.
    \end{equation}
    Lastly, to bound $I_4$, first note that the update law satisfies
    \begin{equation}
        \|\hat{\lambda}_t\| \leq \gamma (1+\bar{l})(1+\bar{k})\,,
    \end{equation}
    and thus applying the bound \eqref{eq:kt-bound-lem}, \eqref{eq:def-M} given in Lemma \ref{lem:ktbound} yields
    \begin{equation} \label{eq:kernelHat-Bound}
        \|\hat{k}_t\| \leq \epsilon + \gamma(1+\bar{l})(1+\bar{k})^2e^{2\bar{\lambda}}\,.
    \end{equation}
    Then, expanding $I_4$ using \eqref{eq:backsteppinglhat} and applying \eqref{eq:kernelHat-Bound} yields
    \begin{equation}
        |I_4(x, t)| \leq 4 (\epsilon(1+\bar{l}) + \gamma(1+\bar{l})^2(1+\bar{k})^2 e^{2\bar{\lambda}}) \|\hat{w}_x\|^2\,.
    \end{equation}
    Combining the bounds on $I_2, I_3, I_4 $ yields the following estimate
    \begin{equation}
        \dot{V} \leq -\frac{4}{1+\|\hat{w}\|^2}\left[\frac{1}{4}-\left(1+\bar{l}\right)\left(3\epsilon - \gamma(1+\bar{l})(1+\bar{k})^2  e^{2\bar{\lambda}}\right)\right] \|\hat{w}_x\|^2\,,
    \end{equation}
    yielding
    \begin{eqnarray}
        \dot{V} &\leq& - \frac{4(1-\gamma/\gamma^*)}{(1+\|\hat{w}\|^2)} \int_0^1 \hat{w}_x^2(x) dx\,, \\ 
        \gamma^* &=& \frac{\frac{1}{4}-3(1+\bar{l})\epsilon}{(1+\bar{l})^2(1+\bar{k})^2}\,.
    \end{eqnarray}
    Thus, $\|w\|$ is bounded and $\|w_x\|^2$ is integrable when $\gamma < \gamma^*$. Similarly, we can bound the time derivative of $\|w_x\|^2$ by the following sequence. First, compute the derivative, then apply integration by parts and substitution of the target system. Then, substitute the bounds from Theorem \ref{thm:operatorApproxOfKernel} along with the inverse transform \eqref{eq:backsteppinglhat} and again perform integration of parts yielding
    \begin{alignat}{3}
        \frac{1}{2} \frac{\partial}{\partial t} \|\hat{w}_x\|^2 &\leq&& \int_0^1 \hat{w}_x \hat{w}_{xt} dx \label{eq:99} \\ 
        &\leq&& \hat{w}_x(1, t) \hat{w}_t(1, t) - \hat{w}_x(0, t) \hat{w}_t(0, t) \nonumber \\ && -&\int_0^1 \hat{w}_{xx}(x, t) \hat{w}_t(x, t) dx \\ 
        &\leq & -&\int_0^1 \hat{w}_{xx}^2(x, t) \nonumber \\ && +& 2\int_0^1 \hat{w}_{xx}(x, t) \hat{k}_x(x,x, t) u(x, t) dx \nonumber \\ && +& \int_0^1 \hat{w}_{xx}(x, t) \int_0^x \delta_{k1}(x, y, t) u(y, t) dy dx \nonumber \\ &&+ &\int_0^1 \hat{w}_{xx}(x, t) \int_0^x \hat{k}_t(x, y, t)u(y, t)dy dx \\ 
        &\leq&& (-1 + 3\epsilon(1+\bar{l}) + \gamma(1+\bar{l})^2(1+\bar{k})^2 e^{2\bar{\lambda}}) \|\hat{w}_x\|^2  \label{eq:whatx-dt-bound} \,.
    \end{alignat}
    Now, since $\|\hat{w}_x\|^2$ is integrable for all time when $\gamma < \gamma^*$, integrating \eqref{eq:whatx-dt-bound} yields that $\|\hat{w}\|$ is bounded. Therefore, by Agmon's inequality, $w(x, t)$ is uniformly bounded for all $t \geq 0$. Furthermore, to show that $u(x, t) \to 0$, performing a similar procedure as in \eqref{eq:99}-\eqref{eq:whatx-dt-bound} yields
    \begin{eqnarray}
        \left|\frac{1}{2} \frac{d}{dt} \|\hat{w}\|^2 \right| &\leq& \|\hat{w}_x\|^2 + 2\bar{\lambda}(1+\bar{l})(1-\bar{k})\|\hat{w}\|^2 \nonumber \\ && + 4(\gamma/\gamma^*) \|\hat{w}_x\|^2\,.
    \end{eqnarray}
    where now the right hand side of this inequality is bounded. Thus, using Barbalat's lemma, we get $\|w\| \to 0$ as $t \to \infty$. Further, using Agmon's inequality we have $w(x, t) \to 0$ for all $x \in [0, 1]$ as $t \to \infty$. Lastly, the result holds for the $u$ system via the boundedness of $\hat{l}$ by $\bar{l}$ and the inverse backstepping transform \eqref{eq:inverse-transform}. Thus, the first part of Theorem \ref{thm:main-result} is complete. 

    For the stability estimate, note that the Lyapunov function \eqref{Eq:lyapunovFunc} satisfies
    \begin{eqnarray}
        \|w(t)\|^2 &\leq& (e^{2V(t)}-1)\,, \label{eq:what-sq-bnd} \\ 
        \|\tilde{\lambda}(t)\|^2 &\leq& 2  \gamma V(t) \leq \gamma(e^{2V(t)} - 1) \,, \label{eq:ltilde-sq-bnd}
    \end{eqnarray}
    for all $t\geq 0$. Further, note that the backstepping transformations with $\hat{k}$ and $\hat{l}$ in \eqref{Eq:backsteppingkhat}, \eqref{eq:backsteppinglhat}, we obtain the following inequalities
    \begin{eqnarray}
        \|u(t)\|^2 &\leq& (1+\bar{l}) \|w(t)\|^2\,, \label{eq:bcks-inequality1} \\ 
        \frac{1}{2} \ln(1+\|w(t)\|^2) &\leq& \frac{1}{2} \|w(t)\|^2  \nonumber \\ &\leq& \frac{1}{2} (1+\bar{k})^2 \|u(t)\|^2\,, \label{eq:bcks-inequality2}
    \end{eqnarray}
    Combining \eqref{eq:what-sq-bnd}, \eqref{eq:ltilde-sq-bnd}, \eqref{eq:bcks-inequality1} leads to 
    \begin{eqnarray}
        \Gamma(t) \leq \max(\gamma, (1+\bar{l})^2) \times (e^{2V(t)} - 1)\,.
    \end{eqnarray}
    Using the inequality \eqref{eq:bcks-inequality2} in combination with \eqref{eq:Gamma-def} and \eqref{Eq:lyapunovFunc} yields
    \begin{eqnarray}
        2V(t) \leq \max\left(\frac{1}{\gamma}, (1+\bar{k})^2 \right) \times \Gamma(t) \,,
    \end{eqnarray}
    resulting in the final stability estimate
    \begin{eqnarray}
        \Gamma(t) &\leq& R\left(e^{\rho \Gamma(0)} - 1\right)\,, \\ 
        R&:=& \max \left(\gamma, (1+\bar{l})^2\right)\,, \\ 
        \rho&:=& \max\left(\frac{1}{\gamma}, (1+\bar{k})^2\right)\,,
    \end{eqnarray}
    for all $t\geq 0$.
\end{proof}

We briefly remind the reader that the assumptions that $\hat{\lambda}, \hat{\lambda}_t \in \Lambda$ (i.e. $\hat{\lambda}, \hat{\lambda}_t$ are Lipschitz and bounded at every $t \geq 0$) are strong and cannot be verified apriori (also seen in \cite{lamarque2024adaptive}), but the assumption that $\hat{k}_t$ is differentiable can be satisfied by choosing a differentiable architecture of the DeepONet, namely via using Sigmoidal activation functions. See \cite{cybenko} for the formal definition of Sigmoidal functions and we note $\sigma(x) = \frac{1}{1+e^{-x}}$ is the most common example. The goal of this paper is not to study the differentiability of the neural operator and as such we do not explore this problem, but leave it for future work. Furthermore, we mention that one can employ the passive identifier approach as in \cite{lamarque2024adaptive} to alleviate the assumption of $\hat{k}_t$. 

Lastly, we briefly mention that the effect of $\epsilon$ on $\gamma$ in \eqref{eq:main-gammastar} is as one would expect: as the neural approximation error is increased, one must choose a smaller update gain since the error due to the neural operator-approximated kernel forces the update law to be more conservative in its rate of update.     

\section{Experimental Simulations} \label{sec:simulations}
\textcolor{blue}{The code for the experiments with exact parameter values is packaged into a Jupyter-Notebook available at \url{https://github.com/lukebhan/NeuralOperatorParabolicAdaptiveControl}.}

For experimental simulation, we consider the plant \eqref{eq:parabolicMain1}, \eqref{eq:parabolicMain2}, \eqref{eq:parabolicMain3} which commonly represents a chemical diffusion process where $u$ is the concentration of a molecule and the reaction kinematics are assumed to be linear with a spatially varying proportional gain \cite{https://doi.org/10.1002/anie.200905513}. In this work, we choose the spatially varying coefficient from a class of functions - namely the Chebyshev polynomials $(\lambda(x) = 25\cos(\gamma(\arccos{x}))+ 25)$ for three reasons. First, the polynomials are bounded and equicontinuous, thus aligning with the compactness requirement of Theorem \ref{thm:operatorApproxOfKernel}. Furthermore, the family of polynomials are challenging functions that yield highly unstable plants and exhibit high parameter sensitivity as small changes in $\gamma$ can drastically change the oscillation of the polynomials (for example, see \cite[Figure 1]{bhan_neural_2023}). Therefore, the neural operator must learn the mapping for vastly different shaped functions including unseen $\lambda$ with shapes differing from the functions used in training.
Lastly, the choice of a sufficiently complex family closely resembles a real-world scenario where one most likely will not know the family of $\lambda$ expected, but may try to use some set of basis functions that can accurately approximate a wide class of potential $\lambda$'s. A reasonable choice for such a set is the Chebyshev polynomials as they are well known for approximating bounded polynomials and thus are one such family to train a NO when $\lambda$ is unknown. 
\newpage

\begin{figure}[H]
    \centering
    \includegraphics{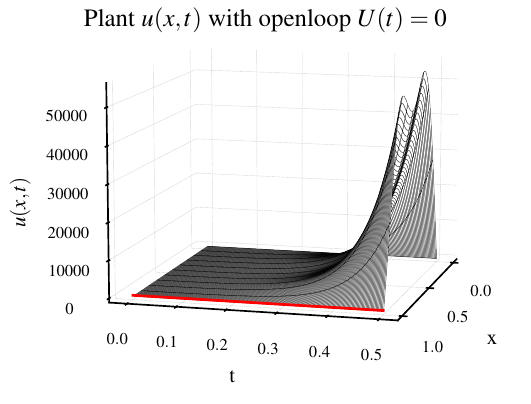}
    \caption{Simulation of the plant \eqref{eq:parabolicMain1}, \eqref{eq:parabolicMain2}, \eqref{eq:parabolicMain3} with openloop controller $U(t) = 0$. We note that the plant is openloop unstable. }
    \label{fig:openloop}
\end{figure}
\begin{figure}[H]
    \centering
    \includegraphics{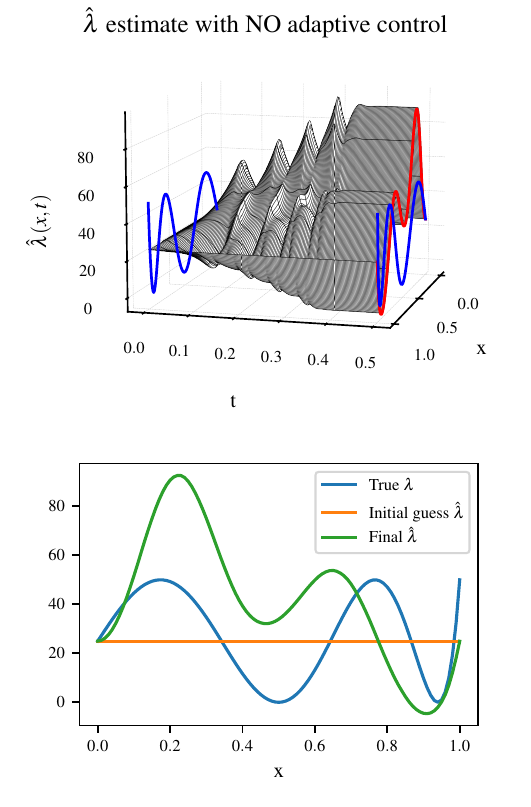}
    \caption{The $\hat{\lambda}$ estimates according to the feedback loop in Figure \ref{fig:plant} (top). \color{black} The blue lines indicate the true $\lambda$ and red line the final $\hat{\lambda}$ estimate.\color{black} The bottom figure shows the final estimates, initial guess, and true $\lambda$ for the simulation. }
    \label{fig:estimate}
\end{figure}

\begin{figure}[H]
    \centering
    \includegraphics{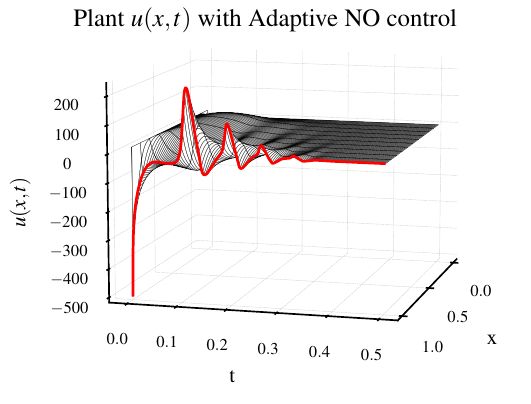}
    \caption{Simulation of the plant \eqref{eq:parabolicMain1}, \eqref{eq:parabolicMain2}, \eqref{eq:parabolicMain3} with the update law \eqref{eq:mainResultUpdateLaw1}, \eqref{eq:mainResultUpdateLaw2},  \eqref{eq:mainResultUpdateLaw3}, and the controller \eqref{eq:finalController} where $\hat{k}$ is calculated using a neural operator.}
    \label{fig:plant}
\end{figure}

\begin{figure}[H]
    \centering
    \includegraphics{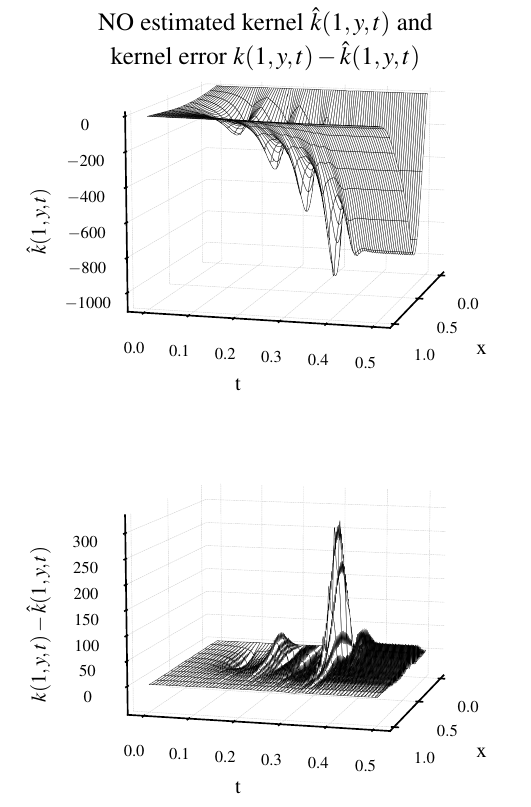}
    \caption{Neural operator approximated kernel $\hat{k}(1, y, t)$ and kernel error corresponding to the plant in Figure \ref{fig:plant}.}
    \label{fig:kernel-error}
\end{figure}

To train a neural operator approximation of the kernels for an unknown set of $\hat{\lambda}$, one must construct a dataset of $\hat{\lambda}$ likely to be seen during implementation. As such, to build our dataset, we choose $10$ $\lambda$ values randomly sampled with \\
$\gamma \sim \text{Uniform}(8.5, 9.5)$ and simulate the trajectories with a finite difference solver (from \cite{pmlr-v242-bhan24a}) for the kernel PDE. We then sample the estimated $\hat{\lambda}$ and $\breve{k}$ at $500$ time-steps for each trajectory consisting of a dataset with $5000$ $(\hat{\lambda}, \breve{k})$ pairs for training and testing the NO. This entire process takes over $1$ hour to simulate trajectories (Nvidia RTX 3090Ti) for only $10$ plants and thus further motivates the use of Neural Operator approximations in the gain kernel. For the NO architecture, we use a DeepONet \cite{lu2021deepxde} where the branch net is a convolutional neural network \cite{726791} and the trunk net is a feed-forward neural-network. 
For the analytical kernel PDE simulation, we use the numerical scheme in the Appendix of \cite{krstic2023neural} and use the benchmark implementation in \cite{pmlr-v242-bhan24a} for the plant PDE.

We present simulations in Figures \ref{fig:openloop}, \ref{fig:kernel-error}, \ref{fig:plant}, \ref{fig:estimate} for a single test case with $\gamma=9$ - a value unseen during training. Th NO approximated kernels are stabilizing, and in practice we found the average (over both space and time) relative error is less than $3\%$ between the operator approximated kernel and the finite difference kernel despite the spikes in error as shown in Figure \ref{fig:kernel-error}. Furthermore, with a spatial step size $dx=0.01$ and a temporal step size of $dt=10^{-5}$ (the largest without numerical instability for the first-order finite difference scheme for the plant), the trajectory simulation (to $T=1$s) took $297$s for the finite difference kernel and $73$s for the NO approximated kernel ($4\times$ speedup). 
This includes the simulation of the PDE plant as well as the calculation of both the update law for the estimator and the controller. For computation of just the gain kernel, we provide calculation times for various spatial step sizes in Table \ref{tab:speedups}. 

\begin{table}[ht]
\centering
\resizebox{0.49\textwidth}{!}{%
\begin{tabular}{|l|l|l|l|}
\hline
\begin{tabular}[c]{@{}l@{}}Spatial step\\ size (dx)\end{tabular} & \begin{tabular}[c]{@{}l@{}}Finite-difference\\ kernel \\ calculation time (ms)\end{tabular} & \begin{tabular}[c]{@{}l@{}}NO \\ kernel \\ calculation time (ms)\end{tabular} & Speedup \\ \hline
0.05                                                             & 0.38                                                                                        & 0.42                                                                          & 0.9$\times$   \\ \hline
0.01                                                             & 0.95                                                                                        & 0.68                                                                          & 13.9$\times$   \\ \hline
0.005                                                            & 39.88                                                                                       & 0.86                                                                          & 45.9$\times$   \\ \hline
\end{tabular}%
}
\caption{Calculation times for the finite difference scheme and neural operator averaged over 100 kernel calculations at various spatial resolutions. }
\label{tab:speedups}
\end{table}

\section{Conclusion}
In this paper, we extend the work of adaptive PDE control with neural operator approximated gains from hyperbolic PDEs~\cite{lamarque2024adaptive} to the more challenging parabolic PDE case. In doing so, we present new Lyapunov results for a more challenging target system \eqref{eq:target} that retains the NO perturbation terms found in \cite{krstic2023neural} coupled with the challenging perturbations from the adaptive scheme in \cite{smyshlyaevBook}. Furthermore, we showcase the power of the neural operator approximation in a scenario where the kernel must be resolved at every timestep achieving speedups of $45\times$ compared to a standard finite-difference implementation. 

\bibliographystyle{elsarticle-num}
\bibliography{references}

\appendix
\section{Wendroff Inequality} \label{appendix:wendroff}
\begin{lemma} (Wendroff Inequality, \cite{Beckenbach_Bellman_1961}). Let $c \geq 0$ and given two non-negative multi-variable functions, $u(r, s)$, $v(r, s) \geq 0$ such that 
\begin{equation}
    u(x, y) \leq c + \int_0^y \int_0^x v(r, s) u(r, s) dr ds\,.
\end{equation}
where $x, y \in \mathbb{R}$. Then the following estimate for $u(x, y)$ holds
\begin{equation}
    |u(x, y)| \leq ce^{\int_0^y \int_0^x v(r, s) u (r, s) dr ds}\, \quad \forall (x, y) \in \mathbb{R}^2\,.
\end{equation}
\end{lemma}

\section{Nonlocal Neural Operator} \label{appendix:nonlocal}
\begin{definition}(Nonlocal Neural Operator (NNO); \cite{lanthaler2023nonlocal}) Let $\Omega_u \subset \mathbb{R}^{d_{u_1}}$, $\Omega_v \subset \mathbb{R}^{d_{v_1}}$ be bounded domains and define the following function spaces consisting of continuous functions $\mathcal{U} \subset C^0(\Omega_u; \mathbb{R}^{d_{u_2}})$, $\mathcal{V} \subset C^0(\Omega_v; \mathbb{R}^{d_{v_2}})$. Then, a NNO is defined as a mapping $\hat{\mathcal{G}} : \mathcal{V} \rightarrow \mathcal{U}$ given by the composition $\hat{\mathcal{G}} = \mathcal{Q} \circ \mathcal{L}_L \circ \cdots \circ \mathcal{L}_1 \circ \mathcal{R}$ consisting of a lifting layer $\mathcal{R}$, hidden layers $\mathcal{L}_l, l=1,..., L$, and a projection layer $\mathcal{Q}$. Furthermore, for $m=0, ..., M$ modes, let there be functions $\psi_{l, m}, \phi_{l, m}: \Omega_u \rightarrow \mathbb{R}^{d_c}$ for a given a channel dimension $d_c > 0$. Then, the lifting layer $\mathcal{R}$ for a function $u \in \mathcal{U}$ is given by
\begin{equation} \label{eq:appendix-lifting-layer}
    \mathcal{R} : \mathcal{U} \rightarrow \mathcal{S}, \quad u(x) \mapsto R(u(x), x)\,, 
\end{equation}
where $\mathcal{S}(\Omega_u; \mathbb{R}^{d_c})$ is a Banach space for the hidden layers and $R: \mathbb{R}^{d_{u_2}} \times \Omega_u \rightarrow \mathbb{R}^{d_c}$ is a learnable neural network acting between finite dimensional Euclidean spaces. For $l=1, ..., L$ each hidden layer $\mathcal{L}_l$ is of the form
\begin{equation} \label{eq:appendix-hiddenLayer}
    (\mathcal{L}_l v)(x) := \sigma \bigg( W_l v(x) + b_l + \sum_{m=0}^M \langle T_{l,m} v, \psi_{l, m} \rangle_{L^2(\Omega_u; \mathbb{R}^{d_c})} \phi_{l, m}(x) \bigg)\,,
\end{equation}
where $W_l, T_{l, m} \in \mathbb{R}^{d_c \times d_c}$ and bias $b_l \in \mathbb{R}^{d_c}$ are learnable parameters, $\sigma: \mathbb{R} \rightarrow \mathbb{R}$ is a smooth, infinitely differentiable activation function that acts component wise on inputs.
 Lastly, the projection layer $\mathcal{Q}$ is defined as 
\begin{equation} \label{eq:appendix-projLayer}
    \mathcal{Q} : \mathcal{S} \rightarrow \mathcal{V}, \quad s(x) \mapsto Q(s(x), x)\,, 
\end{equation}
where $Q$ is a finite dimensional neural network from $\mathbb{R}^{d_c} \times \Omega_u \rightarrow \mathbb{R}^{d_{v_2}}$.
\end{definition}
Note that \eqref{eq:appendix-hiddenLayer} is almost a traditional feed-forward neural network except for the last term that is nonlocal as the inner product is taken over the entire $\Omega_u$ space. 
\subsection{DeepONet form of NNO}
Recall, a DeepONet is of the form
\begin{equation}
    G_\mathbb{N}(\bm{u}_m)(y) = \sum_{k=1}^p g^\mathcal{N}(\bm{u}; \varphi^{(k)})f^\mathcal{N}(y; \theta^{(k)})\,,
\end{equation}
where $\bm{u}_m = (u(x_1), u(x_2), ..., u(x_m))$ for $u \in \mathcal{U}$ and $f^\mathcal{N}$, $g^\mathcal{N}$ are neural networks denote the trunk and branch network respectively. 
To represent the DeepONet as an NNO, let $W_l, b_l = 0$ for every $l=1, ..., L$. Let the branch network $g^\mathcal{N}$ which takes in sensor inputs $u(x_1), ..., u(x_m)$ be given by the following 
\begin{equation}
    g^\mathcal{N}(u) := \sigma\left(\sum_{j=1}^m R(u(x_j), x_j) \right)\,,
\end{equation}
where $R$ is the same neural network in the lifting layer $\mathcal{R}$ in \eqref{eq:appendix-lifting-layer}. Then, let the trunk network $f^\mathcal{N}$ be given by the function $\phi_{l, m}$ in \eqref{eq:appendix-hiddenLayer} with the activation function $\sigma$ and projection layer $\mathcal{Q}$ representing the identity functions. Then, the DeepONet is universal in the setting of Theorem \ref{thm:nnoUniversalApprox}. 
\subsection{FNO form of NNO}
A natural form for the third term in \eqref{eq:appendix-hiddenLayer} is choosing the kernel based operator 
\begin{equation}
    (\mathcal{L}_lv)(x) := \sigma \left(W_lv(x) + b_l + \int_\Omega K_l(x, y) v(y) dy \right).
\end{equation}
Then, one can represent the FNO architecture as a NNO in the following form. Let $K_l(x, y) = K_l(x-y)$ and $K_l(x) = \sum_{|k|\leq k_{\text{max}}} \hat{P}_{l, k} e^{ikx}$ be a trigonometric polynomial (Fourier) approximation with $k_{\text{max}} = M$ modes and $\hat{P}_{l, k}$ is a matrix of complex, learnable parameters, i.e. $\hat{P}_{l, k} \in \mathbb{C}^{d_c \times d_c}$. Then, the FNO is universal in the setting of Theorem \ref{thm:nnoUniversalApprox}. 

\section{Derivation of Target System} \label{appendix:target-system}
We begin by differentiating the transformation
\begin{equation}
    \hat{w}(x, t) = u(x, t) - \int_0^x \hat{k}(x, y, t) u(y, t) dy\,, 
\end{equation}
where we apply integration by parts twice and cancel out the $\hat{k}(x, 0, t) = 0$ and $u(0, t) = 0$ terms yielding
\begin{align} \label{eq:target-wt}
    \hat{w}_t(x, t) =& u_{xx}(x, t) + \lambda(x)u(x, t) - \hat{k}(x, x, t)u_x(x, t) \nonumber \\  &- \int_0^x \hat{k}_{yy}(x, y, t) u(y, t) dy 
   \nonumber \\ & - \int_0^x \left( \hat{k}_t(x, y, t) + \hat{k}(x, y, t)\lambda(y) \right) u(y, t) dy  \,,
\end{align}
and similarly following the application of Leibniz rule for the spatial derivative
\begin{eqnarray} \label{eq:target-wx}
    \hat{w}_{xx}(x, t) &=& u_{xx}(x, t) - 2\hat{k}_x(x, x, t)u(x, t) - \hat{k}(x, x, t)u_x(x, t) \nonumber \\&& - \int_0^x \hat{k}_{xx}(x, y, t)u(y, t) dy \,.
\end{eqnarray}
Now, subtracting \eqref{eq:target-wx} from \eqref{eq:target-wt} and substituting for $k-\hat{k} = \tilde{k}$ and $\lambda - \hat{\lambda} = \tilde{\lambda}$ yields
\begin{eqnarray}
    \hat{w}_t(x, t) - \hat{w}_{xx}(x, t) &=& (\tilde{\lambda}(x, t)-2\tilde{k}_x(x, x, t))u(x, t) \nonumber \\ &&- \int_0^x \bigg[\bigg(\tilde{k}_{xx}(x, y, t) - \tilde{k}_{yy}(x, y, t) \nonumber \\ &&- \lambda(y)\tilde{k}(x, y, t) \bigg) u(y, t) \bigg] dy \nonumber \\ &&  - \int_0^x \tilde{\lambda}(y)k(x, y, t)u(y, t) dy \nonumber \\ && - \int_0^x \hat{k}_t(x, y, t)u(y, t) dy\,.
\end{eqnarray}
And now expanding the $k = \hat{k}+\tilde{k}$ in the second integral yields
\begin{eqnarray}
     \hat{w}_t(x, t) - \hat{w}_{xx}(x, t) &=& (\tilde{\lambda}(x, t) - 2\tilde{k}_{x}(x, x, t))u(x, t) \nonumber \\ &&- \int_0^x \hat{k}(x, y, t)\tilde{\lambda}(y, t)u(y, t) dy \nonumber \\ && 
    \nonumber -\int_0^x \bigg[ \bigg(\tilde{k}_{xx}(x, y, t) - \tilde{k}_{yy}(x, y, t) \\ && - \tilde{k}(x, y, t)\hat{\lambda}(y)\bigg)u(y, t)\bigg]dy  \nonumber \\ && - \int_0^x \hat{k}_t(x, y, t)u(y, t) dy\,.
\end{eqnarray}
Lastly, substituting $\delta_{k0}$ and $\delta_{k1}$ yields the result. 
\end{document}